\documentclass[11pt]{article} 
\usepackage{setspace, ifdraft, blindtext}
\ifdraft{%
    \usepackage{todonotes}}{ 
    \usepackage[disable]{todonotes} 
}
\usepackage{kantlipsum}
\usepackage{tikz}
\usepackage{pgf,tikz}
\usetikzlibrary{arrows}
\usetikzlibrary[patterns]
\usepackage[T1]{fontenc}
\usepackage{enumitem}
\usepackage[english]{babel}
\usepackage{graphicx}

\usepackage{amsmath, amssymb, amsthm}

\theoremstyle{plain}
\newtheorem{theorem}{Theorem}[section]

\newtheorem*{prop*}{Proposition}
\newtheorem{lemma}[theorem]{Lemma}
\newtheorem*{lem*}{Lemma}

\theoremstyle{remark}
\newtheorem*{remark}{Remark}
\theoremstyle{definition}
\newtheorem{definition}[theorem]{Definition}
\makeatletter
\def\th@definition{%
  \thm@notefont{}
  \normalfont 
}
\makeatother

\usepackage[margin=1in]{geometry}
\newcommand{\beq}{\begin{equation}}
\newcommand{\eeq}{\end{equation}}

\usepackage{algorithm}
\usepackage[noend]{algpseudocode}
\usepackage{url}

\usepackage{amsthm,algorithm,graphicx,enumitem,amsmath,mdframed}
\usepackage[noend]{algpseudocode}

\newtheorem{invariant}{Invariant}

\newtheorem{observation}[theorem]{Observation}

\newcommand{\ignore}[1]{}

\def\eps{\epsilon}

\title{Fully Dynamic MIS in Uniformly Sparse Graphs\footnote{A preliminary version of this paper appeared in the proceedings of ICALP'18.}}

\author{Krzysztof Onak\thanks{IBM Research, TJ Watson Research Center, Yorktown Heights, New York, USA}
\and
Baruch Schieber\thanks{IBM Research, TJ Watson Research Center, Yorktown Heights, New York, USA}
\and
Shay Solomon\thanks{IBM Research, TJ Watson Research Center, Yorktown Heights, New York, USA. Supported by the IBM Herman Goldstine Postdoctoral Fellowship.}
\and
Nicole Wein\thanks{Massachusetts Institute of Technology, Cambridge, Massachusetts, USA. Supported by an NSF Graduate Fellowship.}
}

\begin{document}

\maketitle

\begin{abstract}
We consider the problem of maintaining a
maximal independent set (MIS) in a dynamic graph subject to
edge insertions and deletions.
Recently, Assadi, Onak, Schieber and Solomon (STOC 2018) showed that an MIS can be maintained in
\emph{sublinear} (in the dynamically changing number of edges) amortized update time.
In this paper we significantly improve the update time for
\emph{uniformly sparse graphs}.
Specifically, for graphs with arboricity $\alpha$, the amortized update time of our algorithm is $O(\alpha^2 \cdot \log^2 n)$, where $n$ is the number of vertices.
For low arboricity graphs, which include, for example, minor-free graphs as well as some classes of ``real world'' graphs,
our update time is polylogarithmic.
Our update time improves the result of Assadi {\em et al.} for all graphs with arboricity bounded by $m^{3/8 - \eps}$, for any constant $\eps > 0$. This covers much of the range of possible values for arboricity, as the arboricity of a general graph cannot exceed $m^{1/2}$.
\end{abstract}

\section{Introduction}

The importance of the maximal independent set (MIS) problem is hard to overstate.
In general, MIS algorithms constitute a useful subroutine for locally breaking symmetry between several choices.
The MIS problem has intimate connections to a plethora of fundamental combinatorial optimization problems such as maximum matching, minimum vertex cover, and graph coloring.
As a prime example, MIS is often used in the context of graph coloring, as all vertices in an independent set can be assigned the same color.
As another important example, Hopcroft and Karp~\cite{HopcroftKarp} gave an algorithm to compute a large matching (approximating the maximum matching to within a factor arbitrarily close to 1)
by applying maximal independent sets of longer and longer augmenting paths.
The seminal papers of Luby~\cite{Luby86} and Linial~\cite{Linial87} discuss additional applications of MIS.
A non-exhaustive list of further direct and indirect applications of MIS includes resource allocation~\cite{YuWHL14}, leader election~\cite{DaumGKN12}, the construction of network backbones~\cite{KuhnMW04,JurdzinskiK12}, and sublinear-time approximation algorithms~\cite{NguyenO08}.

In the 1980s, questions concerning the computational complexity of the MIS problem spurred a line of research that led to the celebrated parallel algorithms of Luby~\cite{Luby86}, and Alon, Babai, and Itai \cite{AlonBI86}. These algorithms find an MIS in $O(\log n)$ rounds without global coordination. More recently, Fischer and Noever~\cite{FischerN18} gave an even simpler greedy algorithm that considers vertices in random order and takes $O(\log n)$ rounds with high probability (see also an earlier result of Blelloch, Fineman, and Shun~\cite{BlellochFS12}).

In this work we continue the study of the MIS problem by considering the \emph{dynamic setting}, where the underlying graph is not fixed, but rather evolves over time via edge updates.
Formally, a \emph{dynamic graph} is a graph sequence ${\cal G} = (G_0,G_1,\dots,G_M)$ on $n$ fixed vertices,
where the initial graph  is $G_0 = (V,\emptyset)$ and each graph $G_i = (V,E_i)$ is obtained from the previous graph $G_{i-1}$ in the sequence by either adding or deleting a single edge.
The basic goal in this context is to maintain an MIS in time significantly faster than it takes to recompute it from scratch following every edge update.

In STOC'18, Assadi, Onak, Schieber, and Solomon \cite{AOSS2018} gave the first fully dynamic algorithm for maintaining an MIS with sub-linear (amortized) update time.
Their update time is $\min\{m^{3/4},\Delta\}$, where $m$ is the (dynamically changing) number of edges and $\Delta$ is a fixed bound on the maximum degree of the graph. Achieving an update time of $O(\Delta)$ is simple, and the main technical contribution of~\cite{AOSS2018} is in further reducing the update time to $O(m^{3/4})$, which improves over the simple $O(\Delta)$ bound for sufficiently sparse graphs. 

\subsection{Our contribution}
We focus on graphs that are ``uniformly sparse'' or ``sparse everywhere'', as opposed to the previous work by Assadi et al.\ \cite{AOSS2018} that considers unrestricted sparse graphs. We aim to improve the update time of \cite{AOSS2018} as a function of the ``uniform sparsity'' of the graph.
This fundamental property of graphs has been studied in various contexts and under various names over the years, one of which is \emph{arboricity}~\cite{NashW61,NashW64,Tutte61}:
\begin{definition} \label{d1}
The arboricity $\alpha$ of a graph $G=(V,E)$ is defined as $\alpha=\max_{U\subset V} \lceil\frac{|E(U)|}{|U|-1}\rceil$, where $E(U)=\left\{(u,v)\in E\mid u,v\in U\right\}$.
\end{definition}
Thus a graph has bounded arboricity if all its induced subgraphs have bounded density.
The family of bounded arboricity graphs contains, bounded-degree graphs, all minor-closed graph classes (e.g., bounded genus graphs and planar graphs in particular, graphs with bounded treewidth), and randomly generated preferential attachment graphs.
Moreover, it is believed that many real-world graphs such as the world wide web graph and social networks also have bounded arboricity~\cite{GoelG06}.

A dynamic graph of \emph{arboricity} $\alpha$ is a dynamic graph such that all graphs $G_i$ have arboricity bounded by $\alpha$.
We prove the following result.
\begin{theorem} \label{main}
For any dynamic $n$-vertex graph of arboricity $\alpha$, an MIS can be maintained deterministically in $O(\alpha^2\log^2 n)$ amortized update time.
\end{theorem}
Theorem \ref{main} improves the result of Assadi~{\em et al.} for all graphs with arboricity bounded by $m^{3/8 - \eps}$, for any constant $\eps > 0$. This covers much of the range of possible values for arboricity, as the arboricity of a general graph cannot exceed $\sqrt{m}$.
Furthermore, our algorithm has polylogarithmic update time for graphs of polylogarithmic arboricity;
in particular, for the family of constant arboricity graphs the update time is $O(\log^2 n)$.


\subsection{Our and previous techniques}

\subsubsection{Arboricity and the dynamic edge orientation problem} \label{eop}
 Our algorithm utilizes two properties of arboricity $\alpha$ graphs:
 \begin{enumerate}
 \item Every subgraph contains a vertex of degree at most $2\alpha$.
 \item There exists an \emph{$\alpha$ out-degree orientation}, i.e., an orientation of the edges such that every vertex has out-degree at most $\alpha$.
 \end{enumerate}
 The first property follows from Definition \ref{d1} and the second property is due to a well-known alternate definition by Nash-Williams \cite{NashW64}. 

It is known that a bounded out-degree orientation of bounded arboricity graphs can be maintained efficiently in dynamic graphs.
Brodal and Fagerberg~\cite{BrodalF99} initiated the study of the dynamic edge orientation problem and gave an algorithm that maintains an $O(\alpha)$ out-degree
orientation with an amortized update time of $O(\alpha+\log n)$ in any dynamic $n$-vertex graph of arboricity $\alpha$;
 our algorithm uses the algorithm of \cite{BrodalF99} as a \emph{black box}.
(Refer to \cite{Kowalik07,HeTZ14,KopelowitzKPS14,BerglinB17} for additional results on the dynamic edge orientation problem.)

\subsubsection{A comparison to Assadi~{\em et al.}}
As noted already by Assadi~{\em et al.}~\cite{AOSS2018}, a central obstacle in maintaining an MIS in a dynamic graph is to maintain detailed information about the 2-hop neighborhood of each vertex.
Suppose an edge is added to the graph
and as a result, one of its endpoints $v$ is removed from the MIS.
In this case, to maintain the maximality of the MIS, we need to identify and add to the MIS all neighbors of $v$ that are not adjacent to a vertex in the MIS.
This means that for each of $v$'s neighbors, we need to know whether it has a neighbor in the MIS (other than $v$). Dynamically maintaining \emph{complete} information about the 2-hop neighborhood of each vertex
is prohibitive.

To overcome this hurdle, we build upon the approach of Assadi~{\em et al.}~\cite{AOSS2018} and maintain \emph{incomplete} information on the 2-hop neighborhoods of vertices.
Consequently, we may err by adding vertices to the MIS even though they are adjacent to MIS vertices. 
To restore independence, we need to remove vertices from the MIS.
The important property that we maintain, following \cite{AOSS2018}, is that the number of vertices added to the MIS following a single update operation will be
significantly higher than the number of removed vertices.
Like
 \cite{AOSS2018}, we analyze our algorithm using a potential function defined as the number of vertices not in the MIS.
The underlying principle behind achieving a low amortized update time
is to ensure that if the process of updating the MIS takes a long time, then the size of the MIS increases substantially as a result.
On the other hand, the size of the MIS may decrease by at most one following each update. 
Our algorithm deviates significantly from Assadi~{\em et al.} in several respects, described next.

\subparagraph{I: An underlying bounded out-degree orientation.}
We apply the algorithm of \cite{BrodalF99} for efficiently maintaining a bounded out-degree orientation.
Given such an orientation, we can maintain complete information about the \emph{2-hop out-neighborhoods} of vertices (i.e., the 2-hop neighborhoods restricted to the \emph{outgoing edges} of vertices), which helps significantly in the maintenance of the MIS.
More specifically, the usage of a bounded out-degree orientation enables us to reduce the problem to a single nontrivial case, described in Section~\ref{sec:triv}.

\subparagraph{II: An intricate chain reaction.}
To handle the nontrivial case efficiently, we develop an intricate ``chain reaction'' process, initiated by
adding vertices to the MIS that violate the independence property, which then forces the removal of other vertices from the MIS,
which in turn forces the addition of other vertices, and so forth.
Such a process may take a long time. However, by employing a novel partition of a subset of the in-neighborhood of each vertex  into a logarithmic number of ``buckets'' (see point III) together with a careful analysis, we limit this chain reaction to a logarithmic number of steps;
upper bounding the number of steps is crucial for achieving a low update time.
We note that such an intricate treatment was not required
in~\cite{AOSS2018}, where the chain reaction was handled by a simple recursive procedure.

\subparagraph{III: A precise bucketing scheme.}
In order to achieve a sufficient increase in the size of the MIS, we need to carefully choose which vertices to add to the MIS so as to guarantee that at every step of the aforementioned {chain reaction}, even steps far in the future, we will be able to add enough vertices to the MIS. We achieve this by maintaining a precise \emph{bucketing} of a carefully chosen subset of the in-neighborhood of each vertex, where vertices in larger-indexed buckets are capable of setting off longer and more fruitful chain reactions. At the beginning of the chain reaction, we add vertices in the larger-indexed buckets to the MIS, and gradually allow for vertices in smaller-indexed buckets to be added.

\subparagraph{IV: A tentative set of MIS vertices.}
In contrast to the algorithm of Assadi~{\em et al.}, here we cannot iteratively choose which vertices to add to the MIS.
Instead, we need to build a \emph{tentative} set of all vertices that may potentially be added to the MIS, and only later
prune this set to obtain the set of vertices that are actually added there. If we do not select the vertices   added
to the MIS in such a careful manner, we may not achieve a sufficient increase in the size of the MIS. 
To guarantee that the number of vertices actually added to the MIS (after the pruning of the tentative set) is sufficiently large,
we make critical use of the first property of bounded arboricity graphs mentioned in Section \ref{eop}.


\subsubsection{A comparison to other previous work}

Censor-Hillel, Haramaty, and Karnin~\cite{CHK16} consider the problem of maintaining an MIS in the dynamic distributed setting. They show that there is a randomized algorithm that requires only a constant number of rounds in expectation to update the maintained MIS, as long as the sequence of graph updates does not depend on the algorithm's randomness.  
This assumption is often referred to as \emph{oblivious adversary}. 
As noted by Censor-Hillel~{\em et al.}, it is unclear whether their algorithm can be implemented with low total work in the centralized setting. This shortcoming is addressed by Assadi~{\em et al.}~\cite{AOSS2018} and by the current paper.

Kowalik and Kurowski~\cite{KowalikK03} employ a dynamic bounded out-degree
orientation to answer shortest path queries in (unweighted) planar graphs.
Specifically, given a planar graph and a constant $k$, they
maintain a data structure that determines in constant time whether two vertices
are at distance at most $k$, and if so, produces a path of such length between them. This data
structure is fully dynamic with polylogarithmic amortized update time. They show that the case of distance $k$ queries (for a general $k \ge 2$)
 reduces in a clean way to the case of distance 2 queries.
Similarly to our approach, Kowalik and Kurowski maintain information on the 2-hop neighborhoods of
vertices, but the nature of the 2-hop neighborhood information necessary for the two problems is inherently different. For answering distance 2 queries it is required to maintain \emph{complete} information on the 2-hop neighborhoods, whereas for maintaining an MIS one may do with partial information. On the other hand, the 2-hop neighborhood information needed for maintaining an MIS must be more detailed, so as to capture the information \emph{as it pertains to the 1-hop neighborhood}.
Such detailed information is necessary for determining the vertices to be added to the MIS
following the removal of a vertex from the MIS.

\subsection{Dynamic MIS vs.\ dynamic maximal matching}
In the maximal matching problem, the goal is to compute a matching that cannot be extended by adding more edges to it.
The  problem is equivalent to finding an MIS in the line graph of the input graph.
Despite this intimate relationship between the MIS and maximal matching problems, (efficiently) \emph{maintaining} an MIS appears to be inherently harder than maintaining a maximal matching.
As potential evidence,  there is a significant gap in the performance of the naive dynamic algorithms for these problems. (The naive algorithms just maintain for each vertex whether it is matched or in the MIS.)
For the maximal matching problem, the naive algorithm
has an update time of $O(\Delta)$. This is because the naive algorithm just needs to inspect the neighbors of a vertex $v$, when $v$ becomes unmatched, to determine whether it can to be matched.
 For the MIS problem, when a vertex $v$ is removed from the MIS,
 the naive algorithm needs to inspect not only the neighbors of $v$, but also the neighbors of $v$'s neighbors to determine which vertices need to be added to the MIS; as a result the update time is $O(\min\{m,\Delta^2\})$.  
Furthermore, the worst-case number of MIS changes (by any algorithm) may be as large as $\Omega(\Delta)$ \cite{CHK16,AOSS2018},
whereas the worst-case number of changes to the maximal matching maintained by the naive algorithm is $O(1)$.
Lastly, the available body of work on the dynamic MIS problem is significantly sparser than that on the dynamic maximal matching problem.

It is therefore plausible that it may be hard to obtain, for the dynamic MIS problem, a bound better than the best bounds known for the dynamic maximal matching problem.
In particular, the state-of-the-art dynamic \emph{deterministic} algorithm for maintaining a maximal matching has an update time of $O(\sqrt{m})$~\cite{NS13}, even in the amortized sense.
Hence, in order to obtain deterministic update time bounds sub-polynomial in $m$, one may have to exploit the structure of the graph, and bounded arboricity graphs are a natural candidate.
A maximal matching can be maintained in graphs of arboricity bounded by $\alpha$ with amortized update time $O(\alpha +\sqrt{\alpha \log n})$ \cite{NS13,HeTZ14};
as long as the arboricity is polylogarithmic in $n$, the amortized update time is polylogarithmic.
In this work we show that essentially the same picture applies to the seemingly harder problem of dynamic MIS.

\ignore{
\subparagraph{Dynamic MIS vs.\ Dynamic Maximal Matching.} In the Maximal Matching problem, the goal is to compute a matching that cannot be extended by adding another edge. The problem is equivalent to finding an MIS in the line graph of the input graph. In the dynamic setting, a few papers gave

\hrule
\medskip
\hrule

{\bf Dynamic MIS versus dynamic maximal matching}
\begin{enumerate}
\item Surprisingly, there is almost no work on dynamic MIS. The only prior work is the PODC paper by Censor-Hillel et al., which is in the distributed dynamic setting.
\item In contrast, the closely related problem of dynamic maximal matching has been studied a lot in the last decade, citations (in the bibfile) include:
\begin{enumerate}
\item For general graphs, there is a naive algorithm that runs in $O(n)$ worst-case update time.
There's a paper by \cite{IL93} with amortized bounds that improves on the naive alg in a certain regime, but is worse than the naive alg in the complementary regime.
This result was later improved by \cite{NS13} to get $O(\sqrt{m})$ deterministic worst-case update time, which is never worse than the naive alg,
and improves on it for sparse graphs. This $\sqrt{m}$ is the state-of-the art for deterministic algorithms, and even for randomized algorithms with adaptive adversary.
Since the maximal matching problem seems inherently easier than MIS (see next paragraph), it seems that $\sqrt{m}$ is a natural barrier also for MIS in general (sparse graphs). So it's only natural to try and look at uniformly sparse graphs; this might serve as a selling point...
\item For randomized algorithm that assume oblivious adversary, \cite{BGS11} get logarithmic amortized bound, \cite{Sol16} improves this to constant.
(Highly unclear how to get a randomized alg for MIS, I've been thinking on it quite a bit.)
\item For low arboricity graphs, there are deterministic algorithms with amortized \cite{NS13,HTZ14} and worst-case \cite{KKPS14,HTZ14} bounds;
these bounds are polylogarithmic and people try to get it down to a constant...
(There's also a result for maximum matching in trees in a somewhat different model (with logarithmic query time and logarithmic update time), on arxiv by Gupta and Sharma from 2009 \cite{GS09}, we don't have to mention it.)
\item Maximal matching provides 2-approx vertex cover. There are MANY other results for approximate matchings and vertex cover, for general graphs and low arboricity graphs.
You can take refs from my FOCS'16 paper \cite{Sol16} and my SODA'16 paper with Peleg \cite{PS16}, all of them should be in the bibfile.
Dynamic algorithms for maintaining $(2-\eps)$-approx to vertex cover for planar graphs and low arboricity graphs were first studied in my ITCS'18 paper;
for the planar case the results are much better than for low arboricity graphs.
\end{enumerate}
\end{enumerate}

{\bf Why MIS seems inherently more difficult than maximal matching?}
\begin{enumerate}
\item essentially no results here, while plenty there
\item trivial update time here is $m$ (which is a static computation from scratch), while the trivial (worst-case) update time for maximal matching is $n$ (scanning neighbors of updated vertices).
\item  updating a maximal matching requires only one $O(1)$ changes in the worst-case whereas updating an MIS may require many changes, as much as $\Omega(\Delta)$, where $\Delta$ is max degree (e.g. if we have 2 stars and the center of both are in the MIS and then we add an edge to connect the two centers).
\item In maximal matching if you add any edge to a matching you'll need to delete at most two, or if you delete one edge, you'll need to add at most two;
whereas if you add one vertex to MIS you may have to delete $\Delta$ from it, or if you delete one vertex from it, you may need to add $\Delta$.
\item MAIN REASON: for maximal matching we only need information from the 1-hop nbhd whereas here we need info from the 2-hop nbhd.
 And it's inherently more difficult to go 2 hops away; as a concrete evidence, we can mention (not sure it's needed though) that it's a major open question to get $2-\eps$-maximum matching even using randomization, and going 2 hops away can be used to essentially exclude length-3 augmenting paths (the two edges of the augmenting path are the analog of the 2 hops in the MIS problem), and if we could exclude length-3 paths we would get $3/2$-maximum matching (a real breakthrough).
\item For low arboricity graphs low outdegree orientation almost immediately gives rise to polylogarithmic update time for maximal matching (by notifying out-neighbors), whereas for MIS it's unclear how to get anything even for forests.
\item There's also another more implicit issue -- maximum (rather than maximal) independent set is hard to approximate, whereas maximum matching can be solved exactly in $m \sqrt{n}$ time. Not sure what's the situation in low arboricity graphs, though, so we may prefer not to mention this hardness issue...
\end{enumerate}

We should emphasize that the paper is different than the STOC paper and the Kowalik '03 paper.
This is where the "nuggets" should come into play (will send some shortly), and most of this comparison should be made in the "Proof Highlight" subsection.

{\bf Comparison of this paper to Kowalik distance oracle paper '03:}
\begin{enumerate}
\item Similarities: they also consider bounded out-degree orientations and getting info from the 2-hop neighborhood. They have the same trivial cases and the same hard case as us.
\item The above superficial similarities are where the similarities end.
\item Even though both papers consider the 2-hop nbhd, the nature of the information we need from the 2-hop nbhd in MIS is different. In MIS we need info about the 2-hop nbhd \emph{as it pertains} to the 1-hop nbhd. Specifically, when we remove a vertex from the MIS, we need to know which of its nbrs have no nbrs in the MIS so that we can add them to the MIS. Whereas the Kowalik paper only needs to know the set of 2-hop nbrs. Because of this, in MIS if we're considering a vertex v and some vertex in v's 2-hop nbhd changes its MIS status, it is a difficult case when v and this 2-hop nbr have a large common-nbhd. The core of our algorithm is overcoming the obstacle of dealing with vertices with large common nbhds, which is not an obstacle in the distance oracle problem.
\item Dealing with the large common nbhd obstacle. We don't have time to keep track of complete information about the 2-hop nbhd of a vertex as it pertains to the 1-hop nbhd. We only keep track of partial information. In particular when a vertex is removed from the MIS, for some of its nbrs, we don't know whether they need to be added to the MIS. We cannot simply not add these vertices to the MIS bc then it might not be maximal, so instead we add them all and we may get some "false positives", that is we may add some vertices to the MIS that have neighbors in the MIS. We need to remove these nbrs from the MIS, which results in a chain reaction of changes to the MIS. By precisely controlling the chain we show that despite removing these nbrs from the MIS, with each "false positive" vertex we are constantly making progress towards increasing the size of the MIS.
\item In summary, our problem requires more \emph{detailed} 2-hop communication, while Kowalik's requires more \emph{accurate} 2-hop info (in our argument, it is ok to be unsure about whether some vertices need to be added to the MIS as described in previous point) For the Kowalik paper, the requirement of accurate info about 2-hop nbhd results in building and maintaining a "squared graph", a graph structure that could be very dense,
and we can bypass the maintenance of such a dense graph structure via the "false positive" idea. Consequently, while Kowalik's result applies only to minor-free graphs, our result applies to the wider family of bounded arboricity graphs
\end{enumerate}

Nugget(?) about the main invariant: Maintaining the main invariant is problematic bcs things change dynamically. In particular, vertices may leave an active set, leaving "holes" in buckets, and to maintain the invariant we need to close these holes. <- Maybe this is too detailed to mention in the intro bc it requires saying what the main invariant is.

}

\section{Algorithm overview} \label{overview}

Using a bounded out-degree orientation of the edges, a very simple algorithm suffices to handle edge updates that fall into certain cases. The nontrivial case occurs when we remove a vertex $v$ from the MIS and need to determine which vertices in $v$'s in-neighborhood have no neighbors in the MIS, and thus need to be added to the MIS. The in-neighborhood of $v$ could be very large and it would be costly to spend even constant time per in-neighbor of $v$. Furthermore, it would be costly to maintain a data structure that stores for each vertex in the MIS which of its in-neighbors have no other neighbors in the MIS. Suppose we stored such a data structure for a vertex $v$. Then the removal of a vertex $u$ from the MIS could cause the entirety of the common neighborhood of $u$ and $v$ to change their status in $v$'s data structure. If this common neighborhood is large, then this operation is costly.

To address this issue, our algorithm does not even attempt to determine the exact set of neighbors of $v$ that need to be added to the MIS. Instead, we maintain \emph{partial} information about which vertices will need to be added to the MIS. Then, when we are unsure about whether we need to add a specific vertex to the MIS, we simply add it to the MIS and remove its conflicting neighbors from the MIS, which triggers a chain reaction of changes to the MIS. Despite the fact that this chain reaction may take a long time to resolve, we obtain an amortized time bound by using a potential function: the number of vertices not in the MIS. That is, we ensure that if we spend a lot of time processing an edge update, then the size of the MIS increases substantially as a result.

The core of our algorithm is to carefully choose which vertices to add to the MIS at each step of the chain reaction to ensure that the size of the MIS increases sufficiently. To accomplish this, we store an intricate data structure for each vertex which includes a partition of a subset of its in-neighborhood into numbered buckets. The key idea is to ensure that whenever we remove a vertex from the MIS, it has at least one full bucket of vertices, which we add to the MIS.

When we remove a vertex from the MIS whose top bucket is full of vertices, we begin the chain reaction by adding these vertices to the MIS (and removing the conflicting vertices from the MIS). In each subsequent step of the chain reaction, we process more vertices, and for each processed vertex, we add to the MIS the set of vertices in its \emph{topmost full bucket}. To guarantee that every vertex that we process has at least one full bucket, we utilize an invariant (the ``Main Invariant'') which says that for all $i$, when we process a vertex whose bucket $i$ is full, then in the next iteration of the chain reaction we will only process vertices whose bucket $i-1$ is full. This implies that if the number of iterations in the chain reaction is at most the number of buckets, then we only process vertices with at least one full bucket. To bound the number of iterations of the chain reaction, we prove that the number of processed vertices \emph{doubles} at every iteration. This way, there cannot be more than a logarithmic number of iterations. Thus, by choosing the number of buckets to be logarithmic, we only process vertices with at least one full bucket, which results in the desired increase in the size of the MIS.

\section{Algorithm setup}
Our algorithm uses a dynamic edge orientation algorithm as a black box. 
For each vertex $v$, let $N(v)$ denote the neighborhood of $v$, let $N^+(v)$ denote the out-neighborhood of $v$, and let $N^-(v)$ denote the in-neighborhood of $v$.

\subsection{The trivial cases}\label{sec:triv}
Let $\cal{M}$ be the MIS that we maintain.
For certain cases of edge updates, there is a simple algorithm to update the $\cal{M}$. Here, we introduce this simple algorithm and then describe the case that this algorithm does not cover.

\begin{definition} We say that a vertex $v$ is \emph{resolved} if either $v$ is in $\cal M$ or a vertex in $N^-(v)$ is in $\cal{M}$. Otherwise we say that $v$ is \emph{unresolved}.
\end{definition}

The data structure is simply that each vertex $v$ stores (i) the set $M^-(v)$ of $v$'s in-neighbors that are in $\cal{M}$, and (ii) a partition of its in-neighborhood into resolved vertices and unresolved vertices.
To maintain this data structure, whenever a vertex $v$ enters or exits $\cal{M}$, $v$ notifies its 2-hop out-neighborhood. Additionally, following each update to the edge orientation, each affected vertex notifies its 2-hop out-neighborhood.

\begin{mdframed}

\textsc{Delete(u,v)}:
\begin{itemize}
\item It cannot be the case that both $u$ and $v$ are in $\cal M$ since $\cal M$ is an independent set.
\item If neither $u$ nor $v$ is in $\cal M$ then both must already have neighbors in $\cal M$ and we do nothing.
\item If $u\in \cal{M}$ and $v\not\in \cal{M}$, then we may need to add $v$ to $\cal{M}$. If $v$ is resolved, we do not add $v$ to $\cal{M}$. Otherwise, we scan $N^+(v)$ and if no vertex in $N^+(v)$ is in $\cal{M}$, we add $v$ to $\cal{M}$.
\end{itemize}
\textsc{Insert(u,v)}:
\begin{itemize}
\item If it is not the case that both $u$ and $v$ are in $\cal{M}$, then we do nothing and $\cal M$ remains maximal.
\item If both $u$ and $v$ are in $\cal M$ we remove $v$ from $\cal{M}$. Now, some of $v$'s neighbors may need to be added to $\cal{M}$, specifically, those with no neighbors in $\cal{M}$. For each unresolved vertex $w \in N^+(v)$, we scan $N^+(w)$ and if $N^+(w) \cap \cal M=\emptyset$, then we add $w$ to $\cal{M}$. For each resolved vertex $w \in N^-(v)$, we know not to add $w$ to $\cal{M}$. On the other hand, for each unresolved vertex $w \in N^-(v)$, we do not know whether to add $w$ to $\cal M$ and it could be costly to scan $N^+(w)$ for all such $w$. This simple algorithm does not handle the case where $v$ has many unresolved in-neighbors. 
\end{itemize}
\end{mdframed}

In summary, the nontrivial case occurs when we delete a vertex $v$ from $\cal M$ and $v$ has many unresolved in-neighbors.

\subsection{Data structure}
As in the trivial cases, each vertex $v$ maintains $M^-(v)$ and a partition of $N^-(v)$ into resolved vertices and unresolved vertices. In addition, we further refine the set of unresolved vertices. One important subset of the unresolved vertices in $N^-(v)$ is the \emph{active set of v}, denoted $A_v$. As motivated in the algorithm overview, $A_v$ is partitioned into $b$ buckets $A_v(1),\dots,A_v(b)$ each of size at most $s$. We will set $b$ and $s$ so that $b=\Theta(\log n)$ and $s=\Theta(\alpha)$.

The purpose of maintaining $A_v$ is to handle the event that $v$ is removed from $\cal{M}$. When $v$ is removed from $\cal{M}$, we use the partition of $A_v$ into buckets to carefully choose which neighbors of $v$ to add to $\cal{M}$ to begin a chain reaction of changes to $\cal{M}$. For the rest of the vertices in $A_v$, we scan through them and update the data structure to reflect the fact that $v\not\in \cal{M}$. This scan of $A_v$ is why it is important that each active set is small (size $O(\alpha\log n)$).

One important property of active sets is that each vertex is in the active set of at most one vertex. For each vertex $v$, let $a(v)$ denote the vertex whose active set contains $v$. Let $B(v)$ denote the bucket of $A_{a(v)}$ that contains $v$.

For each vertex $v$ the data structure maintains the following partition of $N^-(v)$:
\begin{itemize}
\item $Z_v$ is the set of resolved vertices in $N^-(v)$.
\item $A_v$ (\emph{the active set}) is a subset of the unresolved vertices in $N^-(v)$ partitioned into $b=\Theta(\log n)$ buckets $A_v(1),...,A_v(b)$ each of size at most $s=\Theta(\alpha)$. $A_v$ is empty if $v\not\in\cal{M}$.
\item $P_v$ (\emph{the passive set}) is the set of unresolved vertices in $N^-(v)$ in the active set of some vertex other than $v$. $P_v$ is partitioned into $b$ buckets $P_v(1),...,P_v(b)$ such that each vertex $u\in P_v$ is in the set $P_v(i)$ if and only if $B(u)=i$.
\item $R_v$ (\emph{the residual set}) is the set of unresolved vertices in $N^-(v)$ not in the active set of any vertex.
\end{itemize}

We note that while $Z_v$ depends only on $\cal{M}$ and the orientation of the edges, the other three sets depend on internal choices made by the algorithm. In particular, for each vertex $v$, the algorithm picks at most one vertex $a(v)$ for which $v\in A_{a(v)}$ and this choice uniquely determines for every vertex $u\in N^+(v)$, which set ($A_u$, $P_u$, or $R_u$) $v$ belongs to.

We now outline the purpose of the passive set and the residual set. Suppose a vertex $v$ is removed from $\cal{M}$. We do not need to worry about the vertices in $P_v$ because we know that all of these vertices are in the active set of a vertex in $\cal{M}$ and thus none of them need to be added to $\cal{M}$. On the other hand, we do not know whether the vertices in $R_v$ need to be added to $\cal{M}$. We cannot afford to scan through them all and we cannot risk not adding them since this might cause $\cal{M}$ to not be maximal. Thus, we add them all to $\cal{M}$ and set off a chain reaction of changes to $\cal{M}$. That is, even though our analysis requires that we carefully choose which vertices of $A_v$ to add to $\cal{M}$ during the chain reaction, it suffices to simply add every vertex in $R_v$ to $\cal{M}$ (except for those with edges to other vertices we are adding to $\cal{M}$).

\subsection{Invariants}
We maintain several invariants of the data structure. The invariant most central to the overall argument is the Main Invariant (Invariant~\ref{inv:main}), whose purpose is outlined in the algorithm overview. The first four invariants follow from the definitions.

\begin{invariant}\label{inv:res} \emph{(Resolved Invariant)}. For all resolved vertices $v$, for all $u \in N^+(v)$, $v$ is in $Z_u$. \end{invariant}

\begin{invariant}\label{inv:or} \emph{(Orientation Invariant)}. For all $v$, $Z_v \cup A_v \cup P_v \cup R_v =N^-(v)$. \end{invariant}

\begin{invariant}\label{inv:em} \emph{(Empty Active Set Invariant)}. For all $v \not \in \cal{M}$, $A_v$ is empty. \end{invariant}

\begin{invariant}\label{inv:con} \emph{(Consistency Invariant)}. \vspace{-0.8em}
\begin{itemize}
\item If $v$ is resolved, then for all vertices $u\in N^+(v)$, $v \not \in A_u$.
\item If $v$ is unresolved then $v$ is in $A_u$ for at most one vertex $u\in N^+(v)$.
\item If $v$ is in the active set of some vertex $u$, then for all $w \in N^+(v)\setminus\{u\}$, $v$ is in $P_w(i)$ where $i$ is such that $B(v)=i$.
\item If $v$ is in the residual set for some vertex $u$, then for all $w \in N^+(v)$, $v$ is in $R_w$.
\end{itemize}
\end{invariant}

The next invariant says that the active set of a vertex is filled from lowest bucket to highest bucket and only then is the residual set filled.

\begin{definition}We say that a bucket $A_v(i)$ is \emph{full} if $|A_v(i)|=s$. We say $A_v$ is \emph{full} if all $b$ of its buckets are full. \end{definition}

\begin{invariant}\label{inv:full} \emph{(Full Invariant)}. For all vertices $v$ and all $i<b$, if $A_v(i)$ is not full then $A_v(i+1)$ is empty. Also, if $A_v$ is not full then $R_v$ is empty. \end{invariant}

The next invariant, the Main Invariant, says that if we were to move $v$ from $A_{a(v)}$ to the active set of a different vertex $u$ by placing $v$ in the lowest non-full bucket of $A_u$, then $B(v)$ would not decrease.

\begin{invariant}\label{inv:main} \emph{(Main Invariant)}. For all $v$, if $B(v)=i>1$ then for all $u \in N^+(v) \cap \cal{M}$, $A_u(i-1)$ is full. \end{invariant}

\subsection{Algorithm phases}\label{sec:setup}
The algorithm works in four phases.
\begin{mdframed}
\begin{enumerate}
\item Update $\cal{M}$.
\item Update the data structure.
\item Run a black box dynamic edge orientation algorithm.
\item Update the data structure.
\end{enumerate}
\end{mdframed}

During all phases, for each vertex $v$ we maintain $M^-(v)$. Otherwise, the data structure is completely static during phases 1 and 3.

\section{Algorithm for updating $\cal{M}$}\label{sec:updating_M}


When an edge $(u,v)$ is deleted, we run the procedure $\textsc{Delete(u,v)}$ specified in the trivial cases section. When an edge $(u,v)$ is inserted and it is not the case that both $u$ and $v$ are in $\cal{M}$, then we do nothing and $\cal M$ remains maximal. In the case that both $u$ and $v$ are in $\cal{M}$, we need to remove either $u$ or $v$ from $\cal M$ which may trigger many changes to $\cal{M}$. For the rest of this section we consider this case of an edge insertion $(u,v)$.

The procedure of updating $\cal M$ happens in two stages. In the first stage, we iteratively build two sets of vertices, $S^+$ and $S^-$. Intuitively, $S^+$ is a subset of vertices that we intend to add to $\cal M$ and $S^-$ is the set of vertices that we intend to delete from $\cal{M}$. The aforementioned chain reaction of changes to $\cal{M}$ is captured in the construction of $S^+$ and $S^-$. In the second stage we make changes to $\cal{M}$ according to $S^+$ and $S^-$. In particular, the set of vertices that we add to $\cal{M}$ contains a large subset of $S^+$ as well as some additional vertices, and the set of vertices that we remove from $\cal{M}$ is a subset of $S^-$. In accordance with our goal of increasing the size of $\cal M$ substantially, we ensure that $S^+$ is much larger than $S^-$. 

Why is it important to build $S^+$ before choosing which vertices to add to $\cal{M}$? 
The answer is that it is important that we add a \emph{large} subset of $S^+$ to $\cal{M}$ since our goal is to increase the size of $\cal{M}$ substantially. We find this large subset of $S^+$ by finding a large MIS in the graph induced by $S^+$, which exists (and can be found in linear time) because the graph has bounded arboricity.
Suppose that instead of iteratively building $S^+$, we tried to iteratively add vertices directly to $\cal{M}$ in a greedy fashion. This could result in only very few vertices successfully being added to $\cal{M}$. For example, if we begin by adding the center of a star graph to $\cal{M}$ and subsequently try to add the leaves of the star, we will not succeed in adding any of the leaves to $\cal{M}$. On the other hand, if we first add the vertices of the star to $S^+$ then we can find a large MIS in the star (the leaves) to add it to $\cal{M}$.

\subsection{Stage 1: Constructing $S^+$ and $S^-$}

The crux of the algorithm is the construction of $S^+$ and $S^-$. A key property of the construction is that $S^+$ is considerably larger than $S^-$:

\begin{lemma}\label{lem:bigs+} If $|S^-|>1$ then $|S^+|\geq4 \alpha |S^-|$.\end{lemma}

After constructing $S^+$ and $S^-$ we will add at least $\frac{|S^+|}{2 \alpha}$ vertices to $\cal M$ and remove at most $|S^-|$ vertices from $\cal{M}$. Thus, Lemma~\ref{lem:bigs+} implies that $\cal{M}$ increases by $\Omega(\frac{|S^+|}{\alpha})$.

To construct $S^+$ and $S^-$, we define a recursive procedure \textsc{Process}($w$) which adds at least one full bucket of $A_w$ to $S^+$. A key idea in the analysis is to guarantee that for every call to \textsc{Process}($w$), $A_w$ indeed has at least one full bucket. 


\subsubsection{Algorithm description}

We say that a vertex $w \in S^-$ has been \emph{processed} if \textsc{Process}($w$) has been called and otherwise we say that $w$ is \emph{unprocessed}. We maintain a partition of $S^-$ into the processed set and the unprocessed set and we maintain a partition of the set of unprocessed vertices $w$ into two sets based on whether $A_w$ is full or not. We also maintain a queue $\cal{Q}$ of vertices to process, which is initially empty. Recall that $(u,v)$ is the inserted edge and both $u$ and $v$ are in $\cal{M}$. The algorithm is as follows.

\begin{mdframed}

First, we add $v$ to $S^-$. Then, if $A_v$ is not full, we terminate the construction of $S^+$ and $S^-$. Otherwise, we call \textsc{Process}($v$).

\vspace{.5em}
\noindent\textsc{Process}($w$): \vspace{-0.8em}
\begin{enumerate}
\item If $A_w$ is full, then add all vertices in $A_w(b)\cup R_w$ to $S^+$.
If $A_w$ is not full, then let $i$ be the largest full bucket of $A_w$ and add all vertices in $A_w(i)$ and $A_w(i+1)$ to $S^+$. We will claim that such an $i$ exists (Lemma~\ref{lem:full}).
\item For all vertices $x$ added to $S^+$ in this call to $\textsc{Process}$, we add $N^+(x) \cap \cal M$ to $S^-$.
\item If $S^-$ contains an unprocessed vertex $x$ with full $A_x$, we call \textsc{Process}($x$).
\end{enumerate}

When a call to $\textsc{Process}$ terminates, including the recursive calls, we check whether Lemma~\ref{lem:bigs+} is satisfied (that is, whether $|S^+|\geq4 \alpha |S^-|$), and if so, we terminate. Otherwise, if $\cal{Q}$ is not empty, we let $w$ be the next vertex in $\cal{Q}$ and call \textsc{Process}($w$). If $\cal{Q}$ is empty we enqueue a new batch of vertices to $\cal{Q}$. This batch consists of the set of all unprocessed vertices in $S^-$. We will claim that such vertices exist (Lemma~\ref{lem:full}).
\end{mdframed}


\begin{remark} The reason we terminate without calling \textsc{Process}($v$) if $A_v$ is not full (i.e. $R_v$ is empty) is because $R_v$ is the only set for which we cannot afford to determine whether or not each vertex has another neighbor in $\cal{M}$ (besides $v$): We know that each vertex $w\in Z_v\cup P_v$ has another neighbor in $\cal{M}$, and the set $A_v$ is small enough to scan. For the same reason, step 3 of \textsc{Process} is necessary because it ensures that for every vertex $w$ in $S^-$, all vertices in $R_w$ are in $S^+$. If this weren't the case and we removed a vertex $w$ in $S^-$ from $\cal{M}$, we might be left in the ``hard case'' of needing to deal with $R_w$.
\end{remark}

Lemma~\ref{lem:bigs+} follows from the algorithm specification: either the algorithm terminates immediately with $S^-=\{v\}$ or the algorithm terminates according to the termination condition, which is that Lemma~\ref{lem:bigs+} is satisfied.

Several steps in the algorithm (Step 2 of \textsc{Process}($w$) and the last sentence of the algorithm specification) rely on Lemma~\ref{lem:full}:


\begin{lemma}\label{lem:full}\leavevmode\vspace{-0.8em}
\begin{enumerate}
\item If we call \textsc{Process}($w$), then $A_w$ has at least one full bucket.
\item Every batch of vertices that we enqueue to $\cal{Q}$ is nonempty.
\end{enumerate}
\end{lemma}
\subsubsection{Proof of Lemma~\ref{lem:full}}

Let epoch 1 denote the period of time until the first batch of vertices has been enqueued to $\cal{Q}$. For all $i>1$, let epoch i denote the period of time from the end of epoch $i-1$ to when the $i^{th}$ batch of vertices has been enqueued to $\cal{Q}$.

To prove Lemma~\ref{lem:full}, we prove a collection of lemmas that together show that (i) Lemma~\ref{lem:full} holds for all calls to \textsc{Process} before epoch $b$ ends (recall that $b$ is the number of buckets) and (ii) the algorithm terminates before the end of epoch $b$.

For all $i$, let $p_i$ and $u_i$ be the number of processed and unprocessed vertices in $S^-$ respectively, when epoch $i$ ends. Let $S^+_i$ and $S^-_i$ be the sets $S^+$ and $S^-$ respectively when epoch $i$ ends. Recall that $s$ is the size of a full bucket. Let $s = 8 \alpha$ and let $b = \log_2 n + 1$.

\begin{lemma}\label{lem:fullj} For all $1\leq j\leq b$, every time we call \textsc{Process}($w$) during epoch $j$, $A_w(b-j+1)$ is full. \end{lemma}

\begin{proof} We proceed by induction on $j$.

\noindent {\em Base case.} If $j=1$ then the algorithm only calls \textsc{Process}($w$) on vertices $w$ with full $A_w$ and thus full $A_w(b)$.

\noindent {\em Inductive hypothesis.} Suppose that during epoch $j$, all of the processed vertices have full $A_w(b-j+1)$.

\noindent {\em Inductive step.} We will show that during epoch $j+1$, all of the processed vertices have full $A_w(b-j)$. We first note that during \textsc{Process}($w$), the algorithm only adds the vertices in the topmost full bucket of $A_w$ to $S^+$. Thus, the inductive hypothesis implies that
for all vertices $x \in S^+_j$, $B(x)\geq b-j+1$.

Then, by the Main Invariant, for all $x \in S^+_j$ and all $y \in N^+(x) \cap \cal{M}$, $A_y(b-j)$ is full. By construction, the only vertices in $S^-_j$ other than $v$ are those in $N^+(x) \cap \cal M$ for some $x \in S^+_j$. Thus, 
for all vertices $y \in S^-_j$, $A_y(b-j)$ is full. During epoch $j+1$, the set of vertices that we process consists only of vertices $w$ that are either in $S^-_j$ or have full $A_w$. We have shown that all of these vertices $w$ have full $A_w(b-j)$.
\end{proof}


\begin{lemma}\label{lem:sp} For all $1\leq j\leq b$, $|S^+_j|\geq p_j s$.\end{lemma}

\begin{proof} By Lemma~\ref{lem:fullj}, for all calls to \textsc{Process}($w$) until the end of epoch $j$, $A_w$ has at least one full bucket. During each call to \textsc{Process}($w$), the algorithm adds at least one full bucket (of size $s$) of $A_w$ to $S^+$. By the Consistency Invariant, (i) every vertex is in the active set of at most one vertex and (ii) if a vertex $w$ appears in the active set of some vertex, then $w$ is not in the residual set of any vertex. The only vertices added to $S^+$ are those in some active set or some residual set, so every vertex in some active set that is added to $S^+$, is added at most once. Thus, for each processed vertex, there are at least $s$ distinct vertices in $S^+$.
\end{proof}

\begin{lemma}\label{lem:pu} For all $1\leq j\leq b$, $p_j<u_j$. That is, there are more unprocessed vertices than processed vertices.\end{lemma}

\begin{proof} At the end of epoch $j$, Lemma~\ref{lem:bigs+} is not satisfied because if it were then the algorithm would have terminated. That is, $|S^+_j|<4 \alpha |S^-_j|$. Combining this with Lemma~\ref{lem:sp} and the fact that  $p_j+u_j=|S^-_j|$, we have $p_j s<4 \alpha (p_j+u_j)$. Choosing $s=8 \alpha$ completes the proof.
\end{proof}

\begin{lemma}\label{lem:double} For all $1< j\leq b$, $p_j>2 p_{j-1}$. That is, the number of processed vertices more than doubles during each epoch.\end{lemma}

\begin{proof} At the end of epoch $j-1$, we add all unprocessed vertices to $\cal{Q}$. As a result of calling $\textsc{Process}$ on each vertex in $ \cal Q$, the number of processed vertices increases by $u_{j-1}$ by the end of epoch $j$. That is, $p_j\geq p_{j-1}+u_{j-1}$. By Lemma~\ref{lem:pu}, $p_{j-1}<u_{j-1}$, so $p_j>2 p_{j-1}$.
\end{proof}
We apply these lemmas to complete the proof of Lemma~\ref{lem:full}:
\begin{enumerate}
\item In epoch 1 we process at least one vertex, so $p_1\geq1$. By Lemma~\ref{lem:double}, $p_j>2 p_{j-1}$. Thus, $p_j\geq2^{j-1}$. If $j=b=\log_2 n+1$, then $p_j>n$, a contradiction. Thus, the algorithm never reaches the end of epoch $b$. Then, by Lemma~\ref{lem:fullj}, every time we call \textsc{Process}($w$), $A_w$ has at least one full bucket.
\item Suppose by way of contradiction that we enqueue no vertices to $\cal{Q}$ at the end of some epoch $1\leq j\leq b$. Then, $u_j=0$. By Lemma~\ref{lem:pu}, $p_j<u_j$, so $p_j<0$, a contradiction.
\end{enumerate}

\subsection{Stage 2: Updating $\cal M$ given $S^+$ and $S^-$}\label{sec:MS}

\subsubsection{Algorithm description}

A brief outline for updating $\cal{M}$ given $S^+$ and $S^-$ is as follows: We begin by finding a large MIS $\cal{M'}$ in the graph induced by $S^+$ and adding the vertices in $\cal{M'}$ to $\cal{M}$. This may cause $\cal{M}$ to no longer be an independent set, so we remove from $\cal{M}$ the vertices adjacent to those in $\cal{M'}$. By design, all of these removed vertices are in $S^-$. Now, $\cal{M}$ may no longer be maximal, so we add to $\cal{M}$ the appropriate neighbors of the removed set. The details follow.

\begin{lemma}\label{lem:M} Given $S^+$ and $S^-$, there is an algorithm for updating $\cal M$ so that
\begin{enumerate}
\item at least $\frac{|S^+|}{2 \alpha}$ vertices are added to $\cal M$ and at most $|S^-|$ vertices are deleted from $\cal{M}$.
\item the set of vertices added to $\cal{M}$ is disjoint from the set of vertices removed from $\cal{M}$.
\item after updating, $\cal M$ is a valid MIS.
\end{enumerate}
\end{lemma}

The first step of the algorithm for updating $\cal M$ is to find an MIS $\cal{M'}$ in the graph induced by $S^+$ and add the vertices in $\cal{M'}$ to $\cal{M}$. The following lemma ensures that we can find a sufficiently large $\cal{M'}$.

\begin{lemma}\label{lem:MIS} If $G'$ is a graph on $n'$ nodes of arboricity $\alpha$, then there is an $O(n' \alpha)$ time algorithm to find an MIS $\cal M'$ in $G'$ of size at least $\frac{n'}{2 \alpha}$.\end{lemma}

\begin{proof} Since $G'$ has arboricity $\alpha$, every subset of $G'$ must have a vertex of degree at most $2 \alpha$. The algorithm is simple. Until $G'$ is empty, we repeatedly find a vertex $w$ of degree at most $2 \alpha$, add $w$ to $\cal M'$, and then remove $w$ and $N(w)$ from $G'$. For every vertex we add to $\cal{M}$, we remove at most $2 \alpha$ other vertices from the graph so $|\cal{M'}|\geq$ $ \frac{n'}{2 \alpha}$.

To implement this algorithm, we simply use a data structure that keeps track of the degree of each vertex and maintains a partition of the vertices into two sets: one containing the vertices of degree at most $2\alpha$ and the other containing the rest. Then the removal of each edge takes $O(1)$ time. $G'$ has at most $n' \alpha$ edges so the runtime of the algorithm is $O(n' \alpha)$.
\end{proof}

The algorithm consists of four steps.
\begin{mdframed}
\begin{enumerate}
\item We find an MIS $\cal M'$ in the graph induced by $S^+$ using the algorithm of Lemma~\ref{lem:MIS}. Then we add to $\cal M$ all vertices in $\cal M'$.
\item For all vertices $w \in \cal M'$, we remove from $\cal M$ all vertices in $N^+(w) \cap \cal{M}$. We note that $N^-(w)\cap \cal{M}=\emptyset$ because otherwise $w$ would be resolved, and $S^+$ contains no resolved vertices.
\item 
If $u$ and $v$ are both in $\cal{M}$ (recall that we are considering an edge insertion $(u,v)$), we remove $v$ from $\cal{M}$. We note that this step is necessary because $v$ needs to be removed from $\cal{M}$ and this may not have happened in step 2.
\item Note that for each vertex $w$ removed from $\cal{M}$ in the previous steps, only $R_w$ and some vertices in $A_w$ are in $S^+$. Hence we need to evaluate whether $N^+(w)$ and the remaining vertices in $A_w$ need to enter $\cal{M}$. For each vertex $x \in N^+(w) \cup A_w$, we check whether $x$ has a neighbor in $\cal{M}$. To do this, we scan $N^+(x)$ and if $N^+(x) \cap \cal M=\emptyset$ and $|M^-(x)|=0$, then we add $x$ to $\cal{M}$. (Recall that the data structure maintains $M^-(w)=N^-(w)\cap \cal{M}$.)
\end{enumerate}
\end{mdframed}

\subsubsection{Proof of Lemma~\ref{lem:M}}
While we build $S^+$ and $S^-$, every vertex in $\cal{M}$ that is the out-neighbor of a vertex in $S^+$ is added to $S^-$. The only vertices we remove from $\cal M$ are $v$ and out-neighbors of vertices in $S^+$. Thus, we have the following observation.

\begin{observation}\label{obs} Every vertex that we remove from $\cal{M}$ (including $v$) is in $S^-$.\end{observation}

\begin{enumerate}
\item At least $\frac{|S^+|}{2 \alpha}$ vertices are added to $\cal M$ and at most $| S^-|$ vertices are deleted from $\cal{M}$.
\end{enumerate}

By Lemma~\ref{lem:MIS}, in step 1 we add at least $\frac{|S^+|}{2 \alpha}$ vertices to $\cal{M}$. Observation~\ref{obs} implies that we delete at most $|S^-|$ vertices from $\cal{M}$.

\begin{enumerate}
\item[2.] The set of vertices added to $\cal{M}$ is disjoint from the set of vertices removed from $\cal{M}$.
\end{enumerate}

We will show that the set of vertices added to $\cal{M}$ were not in $\cal{M}$ when the edge update arrived and the set of vertices removed from $\cal{M}$ were in $\cal{M}$ when the edge update arrived.

By Observation~\ref{obs}, every vertex that we remove from $\cal{M}$ is in $S^-$. By construction, a vertex is only added to $S^-$ if it is in $\cal{M}$.

Every vertex that is added to $\cal{M}$ is added in either step 1 or step 4. We claim that every vertex that is added to $\cal{M}$ is the neighbor of a vertex in $S^-$, and thus was not in $\cal{M}$ when the edge update arrived. With Observation~\ref{obs}, this claim is clear for vertices added to $\cal{M}$ during step 4. If a vertex $w$ is added to $\cal{M}$ in step 1 then $w\in S^+$. By construction, a vertex is only added to $S^+$ if it is a neighbor of a vertex in $S^-$.

\begin{enumerate}
\item[3.] $\cal{M}$ is an MIS.
\end{enumerate}
$\cal{M}$ is an MIS because our algorithm is designed so that when we add a vertex $w$ to $\cal{M}$ we remove all of $w$'s neighbors from $\cal{M}$, and when we remove a vertex $w$ from $\cal{M}$ we add to $\cal{M}$ all of $w$'s neighbors that have no neighbors in $\cal{M}$. A formal proof follows.

Suppose by way of contradiction that $\cal M$ is not an independent set. Let $(w, x)$ be an edge such that $w$ and $x$ are both in $\cal{M}$. Due to step 3, $(w, x)$ cannot be $(u, v)$ or $(v, u)$. Since $\cal M$ was an independent set before being updated, at least one of $w$ or $x$ was added to $\cal M$ during the update. Without loss of generality, suppose $w$ was added to $\cal M$ during the update and that the last time $w$ was added to $\cal M$ was after the last time $x$ was added to $\cal{M}$. This means that when $w$ was added to $\cal{M}$, $x$ was already in $\cal{M}$.

\noindent{\bf Case 1:} \emph{w was last added to $\cal M$ during step 4.} In step 4, only vertices with no neighbors in $\cal M$ are added to $\cal M$ so this is impossible.

\noindent{\bf Case 2:} \emph{w was last added to $\cal M$ during step 1 and $x$ was last added to $\cal M$ before the current update to $\cal{M}$.} After $w$ is added to $\cal M$ in step 1, we remove $N^+(v) \cap \cal M$ from $\cal{M}$. Thus, $x \not\in N^+(v)$ so $x$ must be in $N^-(v)$. But $x$ was in $\cal M$ right before the current update so $w$ is resolved. Since $w$ was added to $\cal{M}$ during step 1, $w \in S^+$ and by construction the only vertices added to $S^+$ are those in the active set or residual set of some vertex. By the Consistency Invariant, because $w$ is resolved, it is not in the active set or residual set of any vertex, a contradiction.

\noindent{\bf Case 3:} \emph{$w$ was last added to $\cal M$ during step 1 and $x$ was last added to $\cal M$ during the current update during step 1.} Both $x$ and $w$ are in $S^+$. An MIS in the graph induced by $S^+$ is added to $\cal M$ so $x$ and $w$ could not both have been added to $\cal{M}$, a contradiction.\\

Now, suppose by way of contradiction that $\cal M$ is not maximal. Let $w$ be a vertex not in $\cal M$ such that no vertex in $N(w)$ is in $\cal{M}$. Since $\cal M$ was maximal before the current update, during this update, either $w$ or one of its neighbors was removed from $\cal{M}$. Let $x$ be the last vertex in $N(w) \cup \{w\}$ to be removed from $\cal{M}$. Every vertex that is removed from $\cal M$ (either in step 2 or step 3) has a neighbor in $\cal M$ at the time its removal. Thus, $x$ cannot be $w$ because if $x$ were $w$ then $w$ would be removed from $\cal{M}$ with no neighbors in $\cal{M}$. Thus, $x$ must be in $N(w)$. Since $x$ is removed from $\cal{M}$ during the current update, in step 4 we add to $\cal M$ all vertices in $N^+(x) \cup A_x$ with no neighbors in $\cal{M}$. Thus $w$ is not in $N^+(x) \cup A_x$ so $w$ must be in $Z_x$, $P_x$, or $R_x$.

\noindent{\bf Case 1:} \emph{$w \in Z_x$.} By the definition of resolved, right before the current update, either $w$ was in $\cal{M}$ or there was a vertex $y\in N^-(w)$ in $\cal{M}$.  If $w$ was in $\cal{M}$ then $w$ was removed from $\cal{M}$ during the current update. Every vertex that is removed from $\cal{M}$ (either in step 2 or step 3) has a neighbor that was added to $\cal{M}$ during the current update. By part 2 of the Lemma, the set of vertices that is added to $\cal{M}$ is disjoint from the set of vertices that is removed from $\cal{M}$. Thus, after processing the current update, $w$ has a neighbor in $\cal{M}$, a contradiction.

On the other hand, if $y\in N^-(w)$ was in $\cal{M}$ right before the current update, then during the current update, $y$ was removed from $\cal{M}$. Then in step 4, all vertices in $N^+(y)$ with no neighbors in $\cal M$ are added to $\cal{M}$. This includes $w$, a contradiction.

\noindent{\bf Case 2:} \emph{$w\in P_x$.} Right before the current update, $w$ was in $A_y$ for some $y$. During the current update, $y$ was removed from $\cal{M}$. In step 4, all vertices in $A_y$ with no neighbors in $\cal M$ are added to $\cal{M}$. This includes $w$, a contradiction.

\noindent{\bf Case 3:} \emph{$w\in R_x$.} By Observation~\ref{obs}, every vertex removed from $\cal{M}$ during step 2 is in $S^-$. By construction for all vertices $y \in S^-$, if $R_y$ is nonempty then we add every vertex $R_y$ to $S^+$. Thus, if $x$ was removed from $\cal M$ in step 2, then $w$ was added to $S^+$ in step 1. Also, if $x$ was removed from $\cal M$ in step 3, then $x=v$ and \textsc{Process(x)} was called, which adds $w$ to $S^+$. So, in either case $w$ is in $S^+$. We add an MIS in the graph induced by $S^+$ to $\cal{M}$, so either $w$ or a neighbor of $w$ must be in $\cal{M}$, a contradiction.

\section{Runtime analysis of updating $\cal{M}$}\label{sec:Mrun}

Let $T$ be the amortized update time of the black box dynamic edge orientation algorithm and let $D$ is the out-degree of the orientation. Ultimately, we will apply the algorithm of Brodal and Fagerberg~\cite{BrodalF99} which achieves $D=O(\alpha)$ and $T=O(\alpha+\log n)$. In this section we show that updating $\cal M$ takes amortized time $O(\alpha D)$.

Let $U$ be the total number of updates and let $\Delta^+$ and $\Delta^-$ be the total number of additions of vertices to $\cal{M}$ and removals of vertices from $\cal{M}$, respectively, over the whole computation.

Let $S^+_i$ and $S^-_i$ be the sets $S^+$ and $S^-$, respectively, constructed while processing the $i^{th}$ update (note that this is a different definition of $S^+_i$ than the definition in earlier sections).

\begin{lemma}\label{lem:delta} The amortized number of changes to $\cal M$ is $O(1)$ per update.\end{lemma}

\begin{proof} The proof is based on the observation that each edge update that causes $|\cal{M}|$ to decrease, decreases $|\cal{M}|$ by 1. Thus, on average $|\cal{M}|$ can only increase by at most 1 per update. Then, we apply Lemma~\ref{lem:bigs+} to argue that the total number of changes to $\cal{M}$ per update is within a constant factor of the net change in $|\cal{M}|$ per update. The details follow.

%
%

For every edge deletion, at most 1 vertex is added to $\cal M$ and no vertices are removed from $\cal M$. For every edge insertion, Lemma~\ref{lem:M} implies that at most $|S^-|$ vertices are removed from $\cal{M}$ and at least $\frac{|S^+|}{2 \alpha}$ vertices are added to $\cal{M}$. Thus, the net increase in $|\cal{M}|$ is at least $\frac{|S^+|}{2 \alpha}-|S^-|$, which by Lemma~\ref{lem:bigs+} is at least $|S^-|$ if $|S^-|>1$.
Otherwise, $|S^-|\leq 1$ so the net decrease in $|\cal{M}|$ is at most 1. Thus, the increase in $|\cal{M}|$ over the whole computation, which is equivalent to $\Delta^+-\Delta^-$, is at least $\sum_{i=1}^U |S^-_i|-U$. For every edge update at most $|S^-|$ vertices are removed from $\cal{M}$ (Lemma~\ref{lem:M}), so $\sum_{i=1}^U |S^-_i|\geq \Delta^-$. Thus, we have \begin{equation}\label{eqn:deltas}
\Delta^+-\Delta^-\geq \sum_{i=1}^U |S^-_i|-U \geq \Delta^--U.
\end{equation}

We now show that $\Delta^- = O(U)$ and $\Delta^+ = O(U)$. Because the graph is initially empty, $\cal{M}$ initially contains every vertex. Thus, over the whole computation, the number of additions to $\cal{M}$ cannot exceed the number of removals from $\cal{M}$. That is, $\Delta^+\leq \Delta^-$.

%

Combining this with Equation~\eqref{eqn:deltas}, we have
$\Delta^- = O(U)$ and $\Delta^+ = O(U)$.
\end{proof}

\begin{lemma}\label{lem:sumS}$\sum_{i=1}^U |S^+_i| =O(\alpha U)$.\end{lemma}

\begin{proof} By Lemma~\ref{lem:M}, $\Delta^+\geq\sum_{i=1}^U \frac{|S^+_i|}{2 \alpha}$. By Lemma~\ref{lem:delta}, $\Delta^+= O(U)$.
Thus, $\sum_{i=1}^U |S^+_i| =O(\alpha U)$.
\end{proof}

\begin{lemma} Constructing $S^+$ and $S^-$ takes amortized time $O(\alpha D)$.\end{lemma}

\begin{proof} By construction, it takes constant time to decide to add a single vertex to $S^+$, however, the same vertex could be added to $S^+$ multiple times. A vertex $w$ can be added to $S^+$ only when \textsc{Process}($x$) is called for some $x\in N^+(w)$. \textsc{Process} is called at most once per vertex (per update), so each vertex is added to $S^+$ a maximum $|N^+(w)|\leq D$ times.
 Therefore, over the whole computation it takes a total of $O(D \sum_{i=1}^U |S^+_i|)$ time to build $S^+$.

To build $S^-$, we simply scan the out-neighborhood of every vertex in $S^+$. Thus, over the whole computation it takes $O(D \sum_{i=1}^U |S^+_i|)$ time to build $S^-$.

By Lemma~\ref{lem:sumS}, $\sum_{i=1}^U (|S^+_i|) =O(\alpha U)$ so the amortized time to construct $S^+$ and $S^-$ is $O(\alpha D)$.
\end{proof}

\begin{lemma} Updating $\cal M$ given $S^+$ and $S^-$ takes amortized time $O(D^2\log n)$.\end{lemma}

\begin{proof} 

In step 1 we begin by constructing the graph induced by $S^+$. To do this, it suffices to scan the out-neighborhood of each vertex in $S^+$ since every edge in the graph induced by $S^+$ must be outgoing of some vertex in $S^+$. This takes time $O(|S^+| D)$. Next, we find $\cal M'$, which takes time $O(|S^+|)$ by Lemma~\ref{lem:MIS}.

In step 2, for each vertex $w$ in $\cal M'$, we scan $N^+(v)$, which takes time $O(|S^+| D)$. By Lemma~\ref{lem:sumS}, steps 1 and 2 together take amortized time $O(\alpha D)$.

Step 3 takes constant time.

In step 4, we scan the 2-hop out-neighborhood of each vertex $w$ that has been removed from $\cal{M}$, as well as the out-neighborhood of some vertices in $A_w$. Over the whole computation, this takes time $O(\Delta^-(D^2+\alpha D\log n))$, which is amortized time $O(D^2+\alpha D\log n)$ by Lemma~\ref{lem:delta}.

In total the amortized time is $O(D^2+\alpha D\log n)=O(D^2\log n)$.
\end{proof}

%

\section{Updating the data structure}

\subsection{Satisfying the Main Invariant}

The most interesting part of the runtime analysis involves correcting violations to the Main Invariant. This is also the bottleneck of the runtime. The Main Invariant is violated when a vertex $v$ is added to $\cal{M}$. In this case, $A_v$ is empty and needs to be populated.

We begin by analyzing the runtime of two basic processes that happen while updating the data structure: adding a vertex to some active set (Lemma~\ref{lem:add}) and removing a vertex from some active set (Lemma~\ref{lem:rem}).

Recall that $T$ is the amortized update time of the black box dynamic edge orientation algorithm and $D$ is the out-degree of the orientation.


\begin{lemma}\label{lem:add} Suppose vertex $v$ is not in any active set. Adding $v$ to some active set and updating the data structure accordingly takes time $O(D)$.\end{lemma}

\begin{proof} When we add a vertex $v$ to some $A_u(i)$, for all $w \in N^+(v)$ this could causes a violation to the Consistency Invariant. To remedy this, it suffices to remove $v$ from whichever set it was previously in with respect to $w$ (which is not $A_w$) and add it to $P_w(i)$.
\end{proof}

\begin{lemma}\label{lem:rem} Removing a vertex $v$ from some active set $A_u$ and updating the data structure accordingly takes time $O(D \log n)$.\end{lemma}

\begin{proof} When we remove a vertex $v$ from some $A_u(i)$, this leaves bucket $A_u(i)$ not full so the Full Invariant might be violated. To remedy this, we move a vertex, the \emph{replacement vertex}, from a higher bucket or the residual set to $A_u(i)$.
That is, the replacement vertex $w$ is chosen to be any arbitrary vertex from $R_u\cup A_u(i+1)\cup\dots\cup  A_u(b)\cup P_u(i+1)\cup \dots \cup P_u(b)$. We can choose a vertex from this set in constant time by maintaining a list of all non-empty $P_x(i)$ and $A_x(i)$ for each vertex $x$. If $w$ is chosen from $P_u$, we remove $w$ from $A_{a(w)}$ before adding $w$ to $A_u(i)$.

The removal of $w$ from its previous bucket in its previous active set may leave this bucket not full, so again the Full Invariant might be violated and again we remedy this as described above, which sets off a chain reaction. The chain reaction terminates when either there does not exist a viable replacement vertex or until the replacement vertex comes from the residual set. Since the index of the bucket that we choose the replacement vertex from increases at every step of this process, the length of this chain reaction is at most $b$.\footnote{We note that the length of the chain reaction can be shortened by choosing the replacement vertex from the highest possible bucket. However, we could be forced to choose from bucket $i+1$ (if all buckets higher than $A_u(i+1)$ or  $P_u(i+1)$ are empty).}

For each vertex $v$ that we add to an active set, we have already removed $v$ from its previous active set, so Lemma~\ref{lem:add} applies. Overall, we move at most $b$ vertices to a new bucket and by Lemma~\ref{lem:add}, for each of these $b$ vertices we spend time $O(D)$. Thus, the runtime is $O(b D)=O(D \log n)$.
\end{proof}

\begin{lemma}\label{lem:mtime} The time to update the data structure in response to a violation of the Main Invariant triggered the addition of a single vertex to $\cal M$ is $O(D \alpha \log^2 n)$.\end{lemma}

\begin{proof} To satisfy the Main Invariant, we need to populate $A_v$. We fill $A_v$ in order from bucket 1 to bucket $b$. First, we add the vertices in $R_v$ until either $A_v$ is full or $R_v$ becomes empty. If $R_v$ becomes empty, then we start adding the vertices of $P_v(i)$ in order from $i=b$ to $i=1$; however, we only add vertex $u$ to $A_v$ if this causes $B(u)$ to decrease. Once we reach a vertex $u$ in $P_v$ where moving $u$ to the lowest numbered non-full bucket of $A_v$ does not cause $B(u)$ to decrease, then we stop populating $A_v$. Each time we add a vertex $u$ to $A_v$ from $P_v$, we first remove $u$ from $A_{a(u)}$ and apply Lemma~\ref{lem:rem}. Also, each time we add a vertex to $A_v$, we apply Lemma~\ref{lem:add}. We note that this method of populating $A_v$ is consistent with the Main Invariant.

We add at most $s b=O(\alpha \log n)$ vertices to $A_v$ and for each one we could apply Lemmas~\ref{lem:rem} and \ref{lem:add} in succession. Thus, the total time is $O(\alpha D \log^2 n)$.
\end{proof}

\subsection{Satisfying the rest of the invariants}

Recall the 4 phases of the algorithm from Section~\ref{sec:setup}. We update the data structure in phases 2 and 4.

\subsubsection{Phase 2 of the algorithm}

In phase 2 we need to update the data structure to reflect the changes to $\cal M$ that occur in phase 1. We note that updating $\cal M$ could cause immediate violations to only the Resolved Invariant, the Empty Active Set Invariant, and the Main Invariant. That is, in the process of satisfying these invariants, we will violate other invariants but the above three invariants are the only ones directly violated by changes to $\cal{M}$. We prove one lemma for each of these three invariants (the lemma for the Main Invariant was already proven as Lemma~\ref{lem:mtime}). 

As a technicality, in phase 2 of the algorithm we update the data structure ignoring the newly added edge and instead wait until the data structure update in phase 4 of the algorithm to insert and orient this edge. This means that during the data structure update in phase 2 of the algorithm, $\cal M$ may not be maximal with respect to the graph that we consider.

The time bounds in the following Lemmas (and Lemma~\ref{lem:mtime}) are stated with respect to a single change to $\cal{M}$. By Corollary 10, the amortized number of changes to $\cal M$ is $O(1)$ per update, so the time spent per change to $\cal M$ is the same as the amortized time per update.

\begin{lemma}\label{lem:restime} The time to update the data structure in response to a violation of the Resolved Invariant triggered by a single change to $\cal M$ is $O(D^2 \log n)$.\end{lemma}

\begin{proof} Suppose $v$ is removed from $\cal{M}$. Previously every vertex in $N^+(v)\cup\{v\}$ was resolved and now some of these vertices might be unresolved. To address the Resolved Invariant, for each newly unresolved vertex $u\in N^+(v)\cup\{v\}$, for all $w \in N^+(u)$, we need to remove $u$ from $Z_w$. To accomplish this, we do the following. For all $u \in N^+(v)$, if $|M^-(u)|=0$, then we know that $u$ has become unresolved. In this case, we scan $N^+(u)$ and for each $w \in N^+(u)$, we remove $u$ from $Z_w$. Now, for each such $w$, we need to add $u$ to either $A_w$, $P_w$, or $R_w$. To do this, we find the vertex $x \in N^+(u) \cap \cal M$ with the smallest $|A_x|$. If $A_x$ is not full, then we add $u$ to the lowest bucket of $A_x$ that's not full and apply Lemma~\ref{lem:add}. If $A_x$ is full then for all $w \in N^+(u)$ we add $u$ to $R_w$.

There are at most $D+1$ vertices in $N^+(v)\cup\{v\}$ and for each $u\in N^+(v)$ we spend $O(D)$ time scanning $N^+(u)$. Then, for each $u$, we may apply Lemma~\ref{lem:add}. Thus, the runtime is $O(D^2)$.

Suppose $v$ is added to $\cal{M}$. Some vertices in $N^+(v)\cup\{v\}$ may have previously been unresolved, but now they are all resolved. We scan $N^+(v)\cup\{v\}$ and for all $u \in N^+(v)$ and all $w \in N^+(u)$, we remove $u$ from whichever set it is in with respect to $w$ and add $u$ to $Z_w$. If we removed $u$ from $A_w$, then we update the data structure according to Lemma~\ref{lem:rem}.

There are $D+1$ vertices in $N^+(v)\cup\{v\}$ and each $u\in N^+(v)$ can only be removed from a single $A_w$ because the Consistency Invariant says that each vertex is in at most one active set. Thus, the runtime is $O(D^2 \log n)$.
\end{proof}

\begin{lemma}\label{lem:emtime}The time to update the data structure in response to a violation of the Empty Active Set Invariant triggered by a single change to $\cal M$ is $O(D \alpha \log n)$.\end{lemma}

\begin{proof} The Empty Active Set Invariant can be violated only upon removal of a vertex $v$ from $\cal{M}$. To satisfy the invariant, we need to empty $A_v$. For each vertex $u$ in $A_v$, if $|M^-(u)|>0$, then $u$ is resolved and otherwise $u$ is unresolved. First, we address the case where $u$ is resolved. For each vertex $w \in N^+(u)$ we remove $u$ from whichever set it is in with respect to $w$ (which is not $A_w$ since $u$ is in $A_v$) and add $u$ to $Z_w$.

There are $s b=O(\alpha \log n)$ vertices $u$ in $A_v$ and for each $u$ we spend $O(D)$ time scanning  $N^+(u)$. Thus, the runtime for dealing with the resolved vertices in $A_v$ is $O(D \alpha \log n)$.

Now, we address the case where $u$ in $A_v$ is unresolved. For each $u$ in $A_v$ that is unresolved, and for each $w \in N^+(u)$, we need to add $u$ to either $A_w$, $P_w$, or $R_w$. To do this, we perform the same procedure as in Lemma~\ref{lem:restime}: We find the vertex $x$ in $N^+(u) \cap \cal M$ with the smallest $|A_x|$. If $A_x$ is not full, then we add $u$ to the lowest bucket of $A_x$ that's not full and apply Lemma~\ref{lem:add}. If $A_x$ is full then for all $w \in N^+(u)$ we add $u$ to $R_w$.

There are $s b=O(\alpha \log n)$ vertices $u$ in $A_v$ and for each such $u$ we spend $O(D)$ time scanning $N^+(u)$. Then, for each $u$, we may apply Lemma~\ref{lem:add}. Thus, the runtime for dealing with unresolved vertices in $A_v$ is $O(D \alpha \log n)$.
\end{proof}

%
%

\subsubsection{Phase 4 of the algorithm}

In phase 4 of the algorithm, we need to update the data structure to reflect the edge reorientations that occur in phase 3. If the update was an insertion, we also need to update the data structure to reflect this new edge. We note that reorienting edges could cause immediate violations only to the Resolved Invariant and the Orientation Invariant. That is, in the process of satisfying these invariants, we will violate other invariants but the above two invariants are the only ones directly violated by 
edge reorientations and insertions.
Recall that $T$ is the amortized update time of the black box dynamic edge orientation algorithm.

\begin{lemma}\label{lem:edgetime}The total time to update the data structure in response to edge reorientations and insertions is $O(T D \log n)$.\end{lemma}

\begin{proof} Edge reorientations and insertions can cause violations only to the Resolved Invariant and the Orientation Invariant. Consider an edge $(u,v)$ that is flipped or inserted so that it is now oriented towards $v$. If the edge was flipped, then previously $v$ was in either $Z_u$, $A_u$, $P_u$, or $R_u$, and now we remove $v$ from this set. If $v$ was in $A_u$, then we process the removal of $v$ from $A_u$ according to Lemma~\ref{lem:rem} and add $v$ to the active set of another vertex if possible. To do this, we perform the same procedure as in Lemma~\ref{lem:restime}:
we find the vertex $x \in N^+(u) \cap \cal M$ with the smallest $|A_x|$. If $A_x$ is not full, then we add $u$ to the lowest bucket of $A_x$ that's not full and apply Lemma~\ref{lem:add}. If $A_x$ is full then for all $w \in N^+(u)$ we add $u$ to $R_w$.

In the above procedure we apply Lemma~\ref{lem:rem}, spend $O(D)$ time scanning $N^+(v)$, and apply Lemma~\ref{lem:add}. This takes time $O(D \log n)$.

Before the edge $(u,v)$ was flipped or inserted, $u$ was not in $Z_v$, $A_v$, $P_v$, or $R_v$ and now we need to add $u$ to one of these sets.
\begin{itemize}
\item If $u$ is resolved, we add $u$ to $Z_v$.
\item If $u$ is in an active set and the number of the lowest non-full bucket of $A_v$ is less than $B(u)$, then we remove $u$ from $A_{a(u)}$, applying Lemma~\ref{lem:rem}, and add $u$ to the lowest non-full bucket of $A_v$, applying Lemma~\ref{lem:add}.
\item If $A_v$ is not full and $u$ is not in any active set then we add $u$ to the lowest non-full bucket of $A_v$, applying Lemma~\ref{lem:add}.
\item Otherwise, if $u$ is in the active set of some vertex, we add $u$ to $P_v$ in the appropriate bucket.
\item Otherwise, $u$ is not in any active set and we add $u$ to $R_v$.
\end{itemize}

The slowest of the above options is to apply Lemmas~\ref{lem:rem} and \ref{lem:add} in succession which takes time $O(D \log n)$.

We have shown that updating the data structure takes time $O(D \log n)$ per edge flip, which is amortized time $O(D T \log n)$ per update.
\end{proof}

Combining the runtime analyses of updating $\cal{M}$, running the edge orientation algorithm, and updating the data structure, the total runtime of the algorithm is $O(\alpha D \log^2 n + D^2 \log n + T D \log n)$.
The edge orientation algorithm of Brodal and Fagerberg~\cite{BrodalF99} achieves $D=O(\alpha)$ and $T=O(\alpha+\log n)$. Thus, the runtime of our algorithm is $O(\alpha^2 \log^2 n)$.

\bibliography{icalp18}

\begin{thebibliography}{10}

\bibitem{AlonBI86}
Noga Alon, L{\'{a}}szl{\'{o}} Babai, and Alon Itai.
\newblock A fast and simple randomized parallel algorithm for the maximal
  independent set problem.
\newblock {\em J. Algorithms}, 7(4):567--583, 1986.

\bibitem{AOSS2018}
Sepehr Assadi, Krzysztof Onak, Baruch Schieber, and Shay Solomon.
\newblock Fully dynamic maximal independent set with sublinear update time.
\newblock In {\em Proc. 50th Annual {ACM} {SIGACT} Symposium on Theory of
  Computing, {STOC}}, 2018.

\bibitem{BerglinB17}
Edvin Berglin and Gerth~St{\o}lting Brodal.
\newblock A simple greedy algorithm for dynamic graph orientation.
\newblock In {\em Proc. 28th International Symposium on Algorithms and
  Computation, {ISAAC} 2017, December 9-12, 2017, Phuket, Thailand}, pages
  12:1--12:12, 2017.

\bibitem{BlellochFS12}
Guy~E. Blelloch, Jeremy~T. Fineman, and Julian Shun.
\newblock Greedy sequential maximal independent set and matching are parallel
  on average.
\newblock In {\em Proc. 24th {ACM} Symposium on Parallelism in Algorithms and
  Architectures, {SPAA} 2012, Pittsburgh, PA, USA, June 25-27, 2012}, pages
  308--317, 2012.

\bibitem{BrodalF99}
Gerth~Stølting Brodal and Rolf Fagerberg.
\newblock Dynamic representations of sparse graphs.
\newblock In {\em Proc. 6th International Workshop on Algorithms and Data
  Structures {WADS}}, pages 342--351. Springer-Verlag, 1999.

\bibitem{CHK16}
Keren Censor{-}Hillel, Elad Haramaty, and Zohar~S. Karnin.
\newblock Optimal dynamic distributed {MIS}.
\newblock In {\em Proc. {ACM} Symposium on Principles of Distributed Computing,
  {PODC} 2016, Chicago, IL, USA, July 25-28, 2016}, pages 217--226, 2016.

\bibitem{DaumGKN12}
Sebastian Daum, Seth Gilbert, Fabian Kuhn, and Calvin~C. Newport.
\newblock Leader election in shared spectrum radio networks.
\newblock In {\em Proc. {ACM} Symposium on Principles of Distributed Computing,
  {PODC} 2012, Funchal, Madeira, Portugal, July 16-18, 2012}, pages 215--224,
  2012.

\bibitem{FischerN18}
Manuela Fischer and Andreas Noever.
\newblock Tight analysis of parallel randomized greedy {MIS}.
\newblock In {\em Proc. 29th Annual {ACM-SIAM} Symposium on Discrete
  Algorithms, {SODA} 2018, New Orleans, LA, USA, January 7-10, 2018}, pages
  2152--2160, 2018.

\bibitem{GoelG06}
Gaurav Goel and Jens Gustedt.
\newblock Bounded arboricity to determine the local structure of sparse graphs.
\newblock In {\em Proc. Graph-Theoretic Concepts in Computer Science, 32nd
  International Workshop, {WG}}, pages 159--167, 2006.

\bibitem{HeTZ14}
Meng He, Ganggui Tang, and Norbert Zeh.
\newblock Orienting dynamic graphs, with applications to maximal matchings and
  adjacency queries.
\newblock In {\em Proc. Algorithms and Computation - 25th International
  Symposium, {ISAAC} 2014, Jeonju, Korea, December 15-17, 2014}, pages
  128--140, 2014.

\bibitem{HopcroftKarp}
John~E. Hopcroft and Richard~M. Karp.
\newblock An n\({}^{\mbox{5/2}}\) algorithm for maximum matchings in bipartite
  graphs.
\newblock {\em {SIAM} J. Comput.}, 2(4):225--231, 1973.

\bibitem{JurdzinskiK12}
Tomasz Jurdzinski and Dariusz~R. Kowalski.
\newblock Distributed backbone structure for algorithms in the {SINR} model of
  wireless networks.
\newblock In {\em Proc. Distributed Computing - 26th International Symposium,
  {DISC} 2012, Salvador, Brazil, October 16-18, 2012}, pages 106--120, 2012.

\bibitem{KopelowitzKPS14}
Tsvi Kopelowitz, Robert Krauthgamer, Ely Porat, and Shay Solomon.
\newblock Orienting fully dynamic graphs with worst-case time bounds.
\newblock In {\em Proc. Automata, Languages, and Programming - 41st
  International Colloquium, {ICALP} 2014, Copenhagen, Denmark, July 8-11, 2014,
  Part {II}}, pages 532--543, 2014.

\bibitem{Kowalik07}
Lukasz Kowalik.
\newblock Adjacency queries in dynamic sparse graphs.
\newblock {\em Inf. Process. Lett.}, 102(5):191--195, 2007.

\bibitem{KowalikK03}
Lukasz Kowalik and Maciej Kurowski.
\newblock Short path queries in planar graphs in constant time.
\newblock In {\em Proc. 35th Annual {ACM} Symposium on Theory of Computing,
  {STOC}}, pages 143--148, 2003.

\bibitem{KuhnMW04}
Fabian Kuhn, Thomas Moscibroda, and Roger Wattenhofer.
\newblock Initializing newly deployed ad hoc and sensor networks.
\newblock In {\em Proc. 10th Annual International Conference on Mobile
  Computing and Networking, {MOBICOM} 2004, Philadelphia, PA, USA, September 26
  - October 1, 2004}, pages 260--274, 2004.

\bibitem{Linial87}
Nathan Linial.
\newblock Distributive graph algorithms-global solutions from local data.
\newblock In {\em Proc. 28th Annual Symposium on Foundations of Computer
  Science, {FOCS} 1987, Los Angeles, California, USA, 27-29 October 1987},
  pages 331--335, 1987.

\bibitem{Luby86}
Michael Luby.
\newblock A simple parallel algorithm for the maximal independent set problem.
\newblock {\em {SIAM} J. Comput.}, 15(4):1036--1053, 1986.

\bibitem{NashW61}
Crispin~St.J. Nash-Williams.
\newblock Edge-disjoint spanning trees of finite graphs.
\newblock {\em J. London Math. Soc.}, 36(1):445--–450, 1961.

\bibitem{NashW64}
Crispin~St.J. Nash-Williams.
\newblock Decomposition of finite graphs into forests.
\newblock {\em J. London Math. Soc.}, 39(1):12, 1964.

\bibitem{NS13}
Ofer Neiman and Shay Solomon.
\newblock Simple deterministic algorithms for fully dynamic maximal matching.
\newblock In {\em Proc. Symposium on Theory of Computing Conference, STOC 2013,
  Palo Alto, CA, USA, June 1-4, 2013}, pages 745--754, 2013.

\bibitem{NguyenO08}
Huy~N. Nguyen and Krzysztof Onak.
\newblock Constant-time approximation algorithms via local improvements.
\newblock In {\em Proc. 49th Annual {IEEE} Symposium on Foundations of Computer
  Science, {FOCS} 2008, October 25-28, 2008, Philadelphia, PA, {USA}}, pages
  327--336, 2008.

\bibitem{Tutte61}
William~T. Tutte.
\newblock On the problem of decomposing a graph into $n$ connected factors.
\newblock {\em J. London Math. Soc.}, 36(1):221--230, 1961.

\bibitem{YuWHL14}
Dongxiao Yu, Yuexuan Wang, Qiang{-}Sheng Hua, and Francis C.~M. Lau.
\newblock Distributed ({\(\Delta\)}+1)-coloring in the physical model.
\newblock {\em Theor. Comput. Sci.}, 553:37--56, 2014.

\end{thebibliography}



\end{document}